\documentclass{scrartcl}

\usepackage[T1]{fontenc}
\usepackage[utf8]{inputenc}
\usepackage{textcomp}
\usepackage{mathptmx}

\usepackage{amsmath,amssymb,amsthm}
\usepackage{todonotes}

\usepackage{graphicx}
\usepackage{amsmath,amssymb,amsfonts}
\usepackage{mathtools}

\usepackage{algorithm}
\usepackage{algpseudocode}
\floatname{algorithm}{Listing}

\usepackage[hidelinks]{hyperref}

\hypersetup{
  pdfauthor={Hans-Peter Schröcker, Matthias J. Weber},%
  pdftitle={Guaranteed Collision Detection With Toleranced Motions},%
  pdfcreator={},%
  pdfproducer={},%
}

\newcommand{\RSet}{\mathbb{R}}
\newcommand{\ball}[2]{B_{#2}(#1)}
\newcommand{\orbit}[1]{\mathcal{O}_{#1}}

\newcommand{\distrate}{\varrho}
\newcommand{\SE}[1][3]{\mathrm{SE}(#1)}
\newcommand{\GA}[1][3]{\mathrm{GA}$+$(#1)}
\newcommand{\IP}[2]{\langle #1, #2 \rangle}
\newcommand{\Norm}[1]{\lVert #1 \rVert}
\newcommand*{\defeq}{\coloneqq}

\newcommand{\metdata}{\nu} 
\newcommand{\ballradius}{k}
\DeclareMathOperator{\dist}{dist}

\setkomafont{title}{\normalfont}
\setkomafont{sectioning}{\normalfont\bfseries}
\setkomafont{caption}{\small}
\setcapindent{2em}

\newtheorem{lemma}{Lemma}
\newtheorem{proposition}{Proposition}
\theoremstyle{remark}
\newtheorem{remark}{Remark}

\title{Guaranteed Collision Detection\\With Toleranced Motions}

\author{Hans-Peter Schröcker \quad Matthias J. Weber\\
  Unit Geometry and CAD\\
  University Innsbruck}

\begin{document}

\maketitle

\begin{abstract}
  We present a method for guaranteed collision detection with
  toleranced motions. The basic idea is to consider the motion as a
  curve in the 12-dimensional space of affine displacements, endowed
  with an object-oriented Euclidean metric, and cover it with balls.
  The associated orbits of points, lines, planes and polygons have
  particularly simple shapes that lend themselves well to exact and
  fast collision queries. We present formulas for elementary collision
  tests with these orbit shapes and we suggest an algorithm, based on
  motion subdivision and computation of bounding balls, that can give
  a no-collision guarantee. It allows a robust and efficient
  implementation and parallelization. At hand of several examples
  we explore the asymptotic behavior of the algorithm and compare
  different implementation strategies.


\end{abstract}





\section{Introduction}
\label{sec:introduction}

Collision detection between moving solids \cite{ericson05} is an
ubiquitous issue in robotics and other fields, such as computer
graphics or physics simulation. General-purpose strategies are
available as well as algorithms tailored for specific situations. In
general, any algorithm is a trade-off between reliability, accuracy,
and computation time. In some scenarios, for example in trajectory
planning for high-speed heavy robots, collisions might have
devastating effects. Therefore, it is necessary to exclude them with
certainty. This article presents an efficient algorithm that can give
such a guarantee for collisions between polyhedra under certain
assumptions on the motion.

Typical techniques of collision detection include time discretization
\cite[Section~2.4]{ericson05}, motion linearization, and the use of
bounding volumes \cite[Chapter~4]{ericson05} or hierarchies of
bounding volumes \cite[Chapter~6]{ericson05} and
\cite{hubbard95,dingliana00}. Careful application of discretization
and linearization is critical, as collisions with a short time
interval of interference may be missed. The algorithm we propose in
this article uses hierarchies of bounding volumes but it exhibits some
specialties with respect to their construction. When using bounding
volumes, several potentially contradicting requirements have to be
fulfilled \cite[p.~76]{ericson05}: The computation of the bounding
volumes must be efficient, the representation of moving and fixed
object by not too many bounding volumes has to be accurate, and the
collision queries for bounding volumes should be fast.

In this article, we present ideas for a collision detection algorithm
that satisfies these criteria, at least to a certain extent. The
novelty is to use a ball covering of the motion which we regard as a
curve in the 12-dimensional space of affine displacements. This ball
covering can be converted into a system of simple bounding volumes of
elementary objects (points, lines, planes, and polygons). Unlike other
methods for collision detection based on bounding balls, this requires
only the covering of a one-dimensional object (the motion) as opposed
to the covering of a 3D-solid. The bounding volumes are constructed
from the geometry of the moving object \emph{and} its motion and thus
accurately represent the portion of space occupied by the object as it
moves. It is possible to refine the bounding volumes by subdividing
the motion. With every subdivision step, the number of bounding
volumes doubles while their volume shrinks roughly by the factor
$2^{-3}$. This is much better than for ball coverings of non-curve
like objects in 3D. Furthermore, the worst case complexity of
subdivision is usually not achieved because large parts of motion or
moving objects can be pruned away at a low recursion level.

We expect our algorithm to work particularly well for objects that
lend themselves to polyhedral approximations. Other desirable
properties include:
\begin{itemize}
\item The possibility to report approximate time and
location of potential collisions,
\item simple, fast and robust implementation because of recursive
  function calls without the need of sophisticated data structures,
\item the use of fast and robust primitive procedures which are
  well-known in computational geometry, and
\item the possibility of parallelization.
\end{itemize}
Assuming a robust numerical implementation, collisions up to a
pre-defined accuracy will always be detected correctly and reported
non-collisions are guaranteed.

The ball covering in the space of affine displacements is based on an
underlying object-oriented Euclidean metric
\cite{belta02,chirikjian98} which is derived from a mass distribution
in the space of the moving object. This mass distribution already saw
applications in computational kinematics, for example
\cite{hofer04:_motion_design,nawratil07,schroecker09:_evolving_four_bars}.
One of its important properties, the simplicity of certain orbit
shapes with respect to balls of affine displacements, was observed in
\cite{schroecker05}. It is our aim, to exploit this for guaranteed
collision detection under certain not too restrictive assumptions on
the motion.

Since collision detection is a well-explored research topic, we want
to compare the algorithm we propose with other approaches known from
literature.

There is a number of techniques for detecting collisions or computing
the distance of moving objects by means of covering balls or
hierarchies of covering balls
\cite{quinlan94,hubbard95,greenspan96,martinez98,dingliana00,weller09,weller09b}.
We cover the moving objects by different bounding volumes (portions of
balls and hyperboloids of one or two sheets). However, our
construction is based on a covering of a curve in $\RSet^{12}$ (the
motion) by balls. This latter problem is typically simpler and of
lesser complexity than constructing ball coverings for 3D solids.
Depending on the particular situation, certain curve specifics, for
example piecewise rationality, can be used to efficiently generate
systems of bounding balls. In contrast to collision detection with
covering balls, our algorithm typically requires exactly one quadratic
bounding surface (sphere or hyperboloid) per vertex, edge, and face.
As in case of bounding spheres, efficient collision queries are
possible by evaluation of quadratic functions only.

Several authors suggest bounding volume shapes which are specifically
tailored to the needs of collision detection. Typical examples are the
generalized cylinders of \cite{martinez98} and the s-topes of
\cite{bernabeu01}. Usually, they are derived from some ball covering.
An s-tope is, for example, the convex hull of a finite set of spheres.
These shapes visually resemble our bounding volumes. The crucial
difference is that our construction takes into account the objects'
motion in such a way that it is guaranteed to stay within the volume
\emph{during the whole motion.} In this sense, it resembles the
space-time solids of \cite{hubbard95} or the implementation for jammed
random packing of ellipsoids described in \cite{donev05,donev05b}. A
non-collision guarantee is often possible by means of a coarse
representation. Only if a collision cannot be excluded, the bounding
volumes are adaptively refined. The prize we pay for this exact
collision detection with arbitrary accuracy is the necessity to
re-compute the bounding volumes if the motion changes.

It has to be emphasized that our approach requires no
time-discretization or motion linearization. We work with the motion
as it is analytically defined over a parameter interval. Moreover, our
algorithm should not be confused with algorithms for (approximate)
distance calculation although it can report whether or not the
distance falls beyond a certain threshold during a given motion.
Naturally, our algorithm is considerably slower than many of the
sophisticated existing algorithms but its outcome is of a high
quality. A reported no-collision is guaranteed. Our algorithm should
only be used when such a certification is crucial.

The remainder of this article is organized as follows. After
presenting some preliminaries about the object-oriented metric in
\autoref{sec:distance}, we describe the orbits of geometric primitives
in \autoref{sec:orbits} and their basic collision tests in
\autoref{sec:collision-tests}. In \autoref{sec:bounding-balls} we
present some ideas for constructing bounding balls for motions
(curves in high-dimensional Euclidean spaces) and in
\autoref{sec:collision-detection} we synthesize everything via motion
subdivision into a collision detection algorithm.
\autoref{sec:examples} presents examples and timing results.

\section{Distance between affine displacements}
\label{sec:distance}

Measuring the ``distance'' between two displacements is a fundamental
problem of computational kinematics with many important applications.
It is also afflicted with irresolvable defects such as the
incompatibility of units in angle and lengths and dependence on the
chosen coordinate frames. For our algorithm, we require a Euclidean
metric in the space of Euclidean displacements. Following
\cite{belta02,chirikjian98}, we equip the space $\GA$ of affine
displacements with a left-invariant object-oriented Euclidean metric
and embed the Euclidean motion group $\SE$ into $\GA$. After
identification of $\GA$ with $\RSet^{12}$, the Euclidean motion group
is a sub-variety of dimension six in this space. The distance concept
is based on a positive Borel measure (mass distribution) $\mu$, which
defines the scalar product
\begin{equation}
  \label{eq:1}
  \IP{\alpha}{\beta} = \int \IP{\alpha(x)}{\beta(x)}\;\mathrm{d}\mu
\end{equation}
and the squared distance
\begin{equation}
  \label{eq:2}
  \dist^2(\alpha,\beta) =
  \IP{\alpha-\beta}{\alpha-\beta}
\end{equation}
for any $\alpha,\beta \in \GA$. In practice, the mass distribution is
usually given by a number of feature points $m_1,\ldots,m_n \in
\RSet^3$ with positive weights $w_1,\ldots,w_n \in \RSet_+$ whence the
scalar product \eqref{eq:1} and squared distance \eqref{eq:2} become
\begin{equation*}
  \IP{\alpha}{\beta} = \sum_{i=1}^n w_i \IP{\alpha(m_i)}{\beta(m_i)}
  \quad\text{and}\quad
  \dist^2(\alpha,\beta) = \sum_{i=1}^n w_i \Vert \alpha(m_i) - \beta(m_i)\Vert^2.
\end{equation*}
On the right-hand sides, we use the usual scalar product and norm in
Euclidean three-space $\RSet^3$. In this way, the space $\GA$ is
endowed with a Euclidean metric and it makes sense to speak of ``balls
of affine displacements''.

As usual with distance measures for displacements, there is a degree
of arbitrariness in this definition. However, the moving object will
often suggest a mass distribution in a natural way and we expect that
the algorithm's performance will only marginally depend on the choice
of feature points. Both, angle and distance, depend only on the
barycenter $b$ of the mass distribution and its inertia matrix $J =
(j_{kl})_{k,l \in \{1,2,3\}}$. These
are defined via
\begin{equation}
  \label{eq:3}
  b := \vert \mu \vert^{-1} \int x \;\mathrm{d}\mu,\quad
  j_{kl} := \int x_kx_l \;\mathrm{d}\mu
\end{equation}
where $\vert\mu\vert = \int \mathrm{d}\mu$ is the total mass and $x_k$
is the $k$-th coordinate of $x \in \RSet^3$. In the discrete case,
these formulas become
\begin{equation*}
  \vert\mu\vert = \sum_{i=1}^n w_i,\quad
  b = \vert\mu\vert^{-1} \sum_{i=1}^n w_im_i,\quad
  J = \sum_{i=1}^n w_i^2 m_i m_i^T.
\end{equation*}
The feature points $m_i$ and weights $w_i$ as well as the mass
distribution $\mu$ can always be replaced by equally weighted vertices
of the ellipsoid with equation $(x-b)^T \cdot J \cdot (x-b) = 0$
without changing the metric. This is also true in the continuous case.
In spite of the general definition of inner product and distance in
$\RSet^{12}$, their computation is \emph{always} very efficient.

\section{Orbits of points, lines, planes and polygons}
\label{sec:orbits}

The orbit of a point set $X \subset \RSet^3$ with respect to a set $Y
\subset \RSet^{12}$ of affine displacements is the point set
$\orbit{Y}(X) := \{ \alpha(x) \mid x \in X,\ \alpha \in Y\}$. The
orbits of points, lines, and planes (``fat'' points, lines, and
planes) with respect to a ball $B = \ball{g}{R} \subset \RSet^{12}$
(center $g$, radius $R$) of affine displacements have been computed in
\cite[Section~4]{schroecker05}. When summarizing the results from that
article we always assume that $g$ is the identity. If this is not the
case, the orbit shapes have to be displaced by means of~$g$.

The orbit of a point $x=(x_1,x_2,x_3)$ with respect to
$B$ is a ball $\ball{x}{r}$. Its radius $r$ depends on
the radius $R$ of $B$, the point $x$, and the metric in $\RSet^{12}$.
It is given by the formula
\begin{equation}
  \label{eq:4}
  \distrate^2(x) \defeq
  \frac{r^2}{R^2} =
  \frac{1}{\vert \mu \vert} + \sum_{i=1}^3 \frac{x_i^2}{\mu_i}
\end{equation}
where $\mu_1,\mu_2,\mu_3$ are the eigenvalues of the inertia matrix
\eqref{eq:3}. We call the function $\varrho(x)$ defined by
\eqref{eq:4} the \emph{distortion rate} at $x$. It is constant on
ellipsoids with axes parallel to the eigenvectors of the inertia
matrix and attains its minimum at the barycenter of the mass
distribution. This observation helps to picture the distribution of
distortion rate in space.

For sufficiently small $R$, the orbit of a line $X$ with respect to
$B$ is bounded by a one-sheeted hyperboloid of revolution $\Phi$ with
axis $X$ and the orbit of a plane $\xi$ is bounded by a two sheeted
hyperboloid $\Psi$ with plane of symmetry $\xi$. We consider the orbit
of the line $X$ at first. There exist a parameterization $x(t) = a +
tb$, $\Vert b \Vert = 1$, $t \in \RSet$ of $X$ such that
\begin{equation}
  \label{eq:5}
  \distrate^2(a+tb) = \varrho_0^2 + \frac{t^2}{\metdata^2}.
\end{equation}
We call the real values $\varrho_0$ and $\metdata$ the \emph{metric
  data of $X$.} Denote by $(x,y,z)$ coordinates in an orthonormal
frame with origin $a$ and first basis vector $b$. Then the equation of
$\Phi$ reads
\begin{equation}
  \label{eq:6}
  \Phi\colon \frac{x^2}{1-(\metdata/R)^2} + y^2 + z^2 = R^2\varrho_0^2.
\end{equation}
This describes a hyperboloid of revolution if $0 < R < \metdata$.

Similarly, the plane $\xi$ admits a parameterization $x(s,t) = a + sb_1
+ tb_2$ with orthonormal vectors $b_1$, $b_2$ such that
\begin{equation}
  \label{eq:7}
  \distrate^2(a+sb_1+tb_2) =
  \varrho_0^2 + \frac{s^2}{\metdata_1^2} + \frac{t^2}{\metdata_2^2}.
\end{equation}
The real values $\varrho_0$, $\metdata_1$, and $\metdata_2$ are the
plane's \emph{metric data.} In the orthonormal frame with origin $a$ and
vector basis $(b_1,b_2,b_1 \times b_2)$, the equation of $\Psi$
reads
\begin{equation}
  \label{eq:8}
  \Psi\colon \frac{x^2}{1-(\metdata_1/R)^2} +
             \frac{y^2}{1-(\metdata_2/R)^2} +
             z^2 = R^2\varrho_0^2.
\end{equation}
This is a hyperboloid for $0 < R < \min\{\metdata_1,\metdata_2\}$.

Now we consider the orbit $\orbit{P}$ of a simple polygon $P$ with
respect to the ball $B$ of affine displacements (``fat polygon''). For
sufficiently small radius $R$, it is a solid, bounded by patches of
spheres (vertex orbits), one-sheeted hyperboloids of revolution (edge
orbits), and a two sheeted hyperboloid $\Psi$ (face orbit). In order
to find the boundary curves of these surfaces patches, we consider the
balls $\{ \ball{x}{r} \mid x \in P \}$. Their envelope is precisely
$\orbit{P}$. Thus, the patch on $\Psi$ that actually contributes to
the boundary of $\orbit{P}$ is the projection of $P$ onto $\Psi$ via
the surface normals. Any point $x$ in the polygon's plane $\xi$ gives
rise to two projected images on $\Psi$ which are symmetric with
respect to $\xi$. The patch boundaries on $\Psi$ are the projection of
the polygon edges. The following proposition completely describes the
projected image of~$P$. Its proof is a straightforward calculation.

\begin{proposition}
  \label{prop:1}
  Denote by $\sigma$ the non-uniform scaling
  \begin{equation}
    \label{eq:9}
    (x,y,0) \mapsto \Bigl( \frac{a^2x}{a^2+c^2}, \frac{b^2y}{b^2+c^2}, 0 \Bigr).
  \end{equation}
  The projection of a polygon $P$ in the plane $z = 0$ onto the
  hyperboloid
  \begin{equation}
    \label{eq:10}
    \Psi\colon \frac{x^2}{a^2} + \frac{y^2}{b^2} - \frac{z^2}{c^2} + 1 = 0
  \end{equation}
  via the surface normals consists of those points $(x,y,z)$ whose
  orthographic projections $(x,y,0)$ into $z = 0$ lie in~$\sigma(P)$.
\end{proposition}

Proposition~\ref{prop:1} shows that the boundaries between the
individual surface patches on the fat polygon $\orbit{P}$ are circular and
hyperbolic arcs in planes which are ``vertical'' over $\xi$. This is
illustrated in Figure~\ref{fig:fat-polygon}. The left-hand drawing
shows the original polygon in thick line-style, its scaled copy (the
orthographic back-projection of the patch boundaries) and the relevant
circular and hyperbolic arcs.

\begin{figure}[t]
  \centering
  \begin{minipage}{0.4\linewidth}
    \centering
    \includegraphics{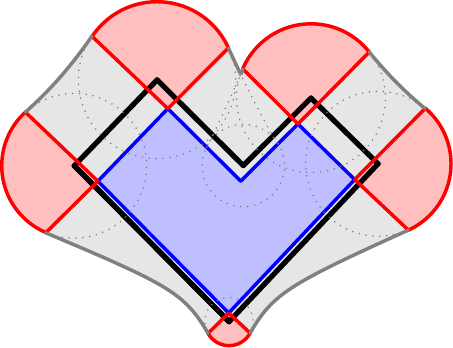}
  \end{minipage}
  \quad
  \begin{minipage}{0.45\linewidth}
    \centering
    \includegraphics{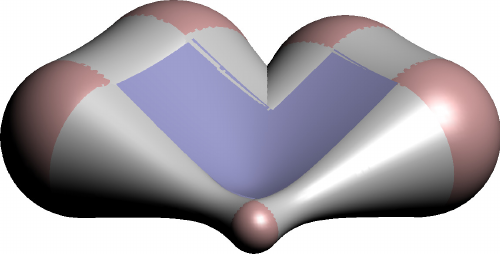}
  \end{minipage}
  \caption{Orbit boundaries for a fat polygon}
  \label{fig:fat-polygon}
\end{figure}

\section{Basic collision tests}
\label{sec:collision-tests}

The orbits of points, edges and faces with respect to balls of affine
displacements are solids bounded by pieces of spheres, hyperboloids of
revolutions and hyperboloids of two sheets. In particular, their
boundary is composed of piecewise quadratic surfaces and allows to
efficiently test for the basic collisions moving vertex vs. fixed
face, moving edge vs. fixed edge and moving face vs. fixed vertex
(\cite[Section~9.1]{ericson05} and \cite{moore88}). Note that we need
not care about ``degenerate'' collisions (for example edge contained
in face) because we use fat points, edges and faces. It is important
that the basic collision tests are implemented efficiently and
reliably, and that numeric issues must not result in an invalid
``no-collision'' response. Due to their simple nature, this is possible.

\subsection*{Moving vertex intersects fixed face}

Via tolerancing, the moving vertices are converted into balls so that
we actually have well-known sphere vs. polygon collision queries
\cite[Section 5.2.8]{ericson05}. These are especially efficient and
simple to implement, if the fixed polygon $F$ is convex.

\subsection*{Moving edge intersects fixed edge}

Here we have to test whether a line segment intersects a hyperboloid
of revolution between a pair of parallel circles. Note that we do not
have to test collision of the line segment with the bounding balls at
either end of the moving edge because these are already handled by the
preceding test. The intersection problem can be converted into a test
whether a quadratic polynomial has zeros in a given interval: Denote
by $a$ and $b$ the intersection points of the line segment with the
two parallel planes. (Should the original line segment be parallel to
the two planes, we have to test whether a similar function has real
roots. This is even simpler.) We parameterize the line segment with end
points $a$ and $b$ by $\ell(t) = (1-t)a + tb$, $t \in [0,1]$. Plugging
$\ell(t)$ into the algebraic equation \eqref{eq:6} of the hyperboloid
of revolution, we obtain a quadratic equation $q(t)$. The exact test
whether it has roots in $[0,1]$ is a trivialization of more advanced
algorithms \cite{rouillier04}.

\subsection*{Moving face intersects fixed vertex}

We have to test whether a point $p$ lies between the two sheets of a
hyperboloid and whether the orthographic projection of this point lies
in the scaled copy $\sigma(P)$ of the moving polyhedron. Here,
$\sigma$ is the non-uniform scaling of Equation~\eqref{eq:9}. These
two tests amount to determining the sign of a quadratic function at
$p$ and a point-in-polygon test \cite[Section 5.4.1]{ericson05}.
Again, the hyperboloids and balls for edges and vertices can be
ignored because their intersections are already handled by preceding
tests.

\section{Bounding balls for motions}
\label{sec:bounding-balls}

Our collision detection algorithm also requires the computation of a
bounding ball for a motion $c(t)$, considered as a curve in
$\SE\subset\GA=\RSet^{12}$. This space is equipped with the
object-oriented Euclidean metric derived from a mass distribution
$\mu$ in the space of the moving polygon. The problem of finding the
(unique) minimal enclosing ball for an arbitrary curve $c(t)$ in a
Euclidean space is hard. However, the requirements of our algorithm
allow simplifications that make the problem not only tractable but
allow efficient solutions in many relevant cases:
\begin{itemize}
\item We are not looking for the optimal bounding ball. A ``good''
  approximation is acceptable as well.
\item We are allowed to subdivide the curve $c(t)$ and compute
  bounding balls for the individual segments. If the curvature and
  arc-length of $c(t)$ can be bounded, after a sufficient number of
  subdivisions one would assume that the ball over the diameter
  defined by the curve end-points will bound the complete curve.
  Below, in \autoref{prop:2} we prove a slightly weaker statement. It
  gives bounding balls whose radius asymptotically equals half the
  curve length. For typical motions in robotics, bounds for arc-length
  and curvature can be obtained if such bounds are known for the joint
  angles.
\item For the important class of piecewise rational motions, the
  convex hull property allows to use the control polygon for efficient
  bounding spheres computation.
\end{itemize}

Here is a simple method to construct bounding balls for short curve
segments with small curvature.

\begin{proposition}
  \label{prop:2}
  Let $C$ be a $C^2$ curve parameterized by $c\colon [t_0,t_1] \to
  \RSet^d$ with arc-length $\ell < r\frac{\pi}{2}$ and curvature
  $\varkappa$ bounded by $\vert\varkappa\vert \le 1/r$. Define $t
  \defeq \ell/r$ and $e_1 \defeq \dot{c}(t_0)/\Norm{\dot{c}(t_0)}$.
  Then the complete curve is contained in the ball
  $\ball{m}{\ballradius}$ with radius $\ballradius = r(1-\cos t)/\sin
  t$ and center $m = c(t_0) + \ballradius e_1$.
\end{proposition}

The proof of \autoref{prop:2} is based on \cite[Lemma~5]{schroecker05}
which we repeat here for convenience of the reader:

\begin{lemma}
  \label{lem:1}
  Assume that $I \subseteq (-r\frac{\pi}{2},r\frac{\pi}{2})$ is an
  interval with $0 \in I$ and $c\colon I \to \RSet^d$ is an arc length
  parameterization of a $C^2$ curve $C$ with curvature $\varkappa$
  such that $\vert\varkappa\vert \le 1/r$. Then there is a
  parameterization $\overline{c}(u) = ue_1 + f(u) \cdot n(u)$ of $C$
  such that $e_1= \dot{c}(0)$, $n(u)$ is a unit vector orthogonal to
  $e_1$, $\vert u(t) \vert \ge r \sin(t/r)$ and $\vert f(u) \vert \le
  r - \sqrt{r^2-u^2}$.
\end{lemma}

\begin{figure}
    \centering
    \includegraphics{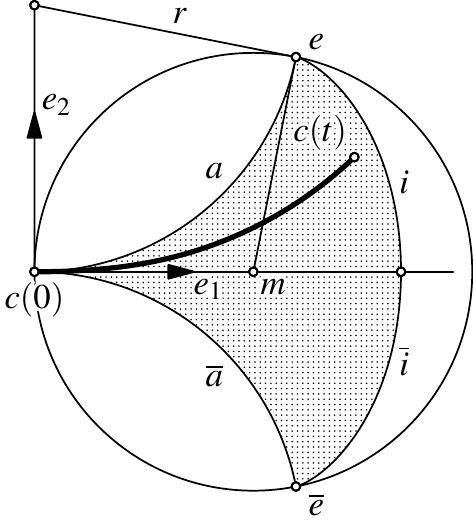}
    \caption{Bounding ball for curve $c(t)$ with bounded arc
      length and curvature.}
    \label{fig:involute}
\end{figure}

This lemma gives us a bounding volume for the curve $C$:
\begin{lemma}
  \label{lem:2}
  We use the same notation and assumptions as in \autoref{lem:1}. In
  case $d=2$, the curve $C^\star$ parameterized by the restriction of
  $c$ to $[0,t]$ is contained in the area bounded by
  \begin{itemize}
  \item the circular arcs $a$, $\overline{a}$ of radius $r$ and length
    $t$ which start at $c(0)$ and share the oriented tangent with
    $C^\star$ in this point and
  \item the segments $i$, $\overline{i}$ of involutes to these
    circular arcs which connect the point $c(0)+te_1$ with the arc end
    points $e$,~$\overline{e}$
  \end{itemize}
  (\autoref{fig:involute}). In case of $d > 2$, the curve $C^\star$ is
  contained in the volume obtained by rotating this area about the
  curve tangent in~$c(0)$.
\end{lemma}

\begin{proof}
  We consider the case $d = 2$ first. \autoref{lem:1} implies that the
  curve point $c(t)$ lies in the exterior of the two circles through
  $a$, $\overline{a}$. Because of the bounded arc length $t$, the
  curve will not go beyond the circular involutes $i$, $\overline{i}$.
  If $d > 2$, the claim follows by considering the graph of the
  function $f(u)$ in the plane spanned by $e_1$ and $e_2 = n(0)$.
\end{proof}

\begin{proof}[Proof of \autoref{prop:2}]
  The curve $C$ has an arc-length parameterization that satisfies the
  criteria of \autoref{lem:1}. Hence, we can enclose it in the volume
  given by \autoref{lem:2}. In order to prove the proposition, we have
  to show that the ball $\ball{m}{k}$ contains this volume. Center $m$
  and radius $k$ are chosen such that the start points $c(t_0)$, $e$
  and $\overline{e}$ of curve and involutes lie on the ball boundary.
  It is a little tedious but possible to verify that it also contains
  the involute arcs $i$,~$\overline{i}$ for any value $\ell \in [0,
  r\frac{\pi}{2}]$ (\autoref{fig:involute}).
\end{proof}

\section{Collision detection algorithm}
\label{sec:collision-detection}

Now we sketch an algorithm for collision detection by means of
toleranced motions. We describe only the case of collisions between a
fixed polygon $F$ and a moving polygon $P$, undergoing a motion
$c(t)$. These polygons are assumed to be simple and we think of them
as being composed of vertices, edges, and one face. For the actual
implementation, it is advantageous to have access to lists of vertices
and edges. Collision detection for the case of polyhedra can be
reduced to this situation.

The basic idea of motion tolerancing is to cover a given
one-parametric motion by a sequence of full-dimensional motions and
then test the orbits of geometric entities (balls, straight-line
segments, polygons) with respect to these full-dimensional motions for
collision. If no collision occurs, it can also be excluded for the
originally given motion. If this is not the case, the sequence of
covering motions needs to be refined.

The motion $c(t)$ is a curve on $\SE \subset \GA = \RSet^{12}$,
defined over an interval $I = [t_0,t_1]$. We assume that we can
compute a ``small'' bounding ball for the restriction of $c(t)$ to any
sub-interval of $I$. Here, the words ``ball'' and ``small'' refer to
the object-oriented Euclidean metric in $\RSet^{12}$, induced via a
positive Borel measure $\mu$ in the space of the moving polygon
(\autoref{sec:distance}).

We start our algorithm by a pre-processing step where we compute for
every edge $e_i$ of the moving polygon $P$ and for the plane $\xi$ of
$P$ the metric data as defined by \eqref{eq:5} and \eqref{eq:7}. This
gives us a pair $(\metdata_i,\varrho_{0,i})$ for every edge and a
triple $(\metdata'_1,\metdata'_2,\rho_0)$ for~$\xi$. From a
computational point of view, this amounts to the principal axis
transform of a one- or two-dimensional quadric. Transcendental
functions are involved only at this stage of the algorithm. Since the
metric data depends only on the moving polygon $P$ and the Borel
measure $\mu$, it can be re-used for different motions $c(t)$ and
different fixed polygons~$F$.

In the algorithm itself, we construct a bounding ball $\ball{g}{R}$ of
$c(t)$ and test the toleranced moving polygon $P$ with respect to this
ball for collision with the fixed polygon $F$. The necessary tests are
described in Section~\ref{sec:collision-tests}. We perform them with
respect to a certain ``zero-position'' of $P$ in the plane with
equation $z = 0$. This requires the prior transformation of $F$ via
$g^{-1}$. In return for the necessary inversion of a $3 \times 3$
matrix we get an easier and cleaner implementation. If the orbit of
$P$ with respect to $\ball{g}{R}$ intersects $F$, we recursively
subdivide the motion $c(t)$ and repeat the collision test. If a
``no-collision'' predicate cannot be given before the termination
criterion, the algorithm returns a time interval of possible
collisions. Otherwise, if no collision occurs, it returns an empty
interval. It is possible to add information on the location and type
of collision. Moreover, the sub-tasks obtained from the repeated
subdivision of $c(t)$ can be solved in parallel.

We suggest to terminate this recursive algorithm as soon as the
product
\begin{equation}
  \label{eq:11}
  R_{\min} \defeq R\max_{x \in P}\varrho(x)
\end{equation}
of current bounding ball radius $R$ and maximal distortion rate
attained inside the moving polygon is smaller than a pre-defined value
$\varepsilon > 0$. Pre-computing this value amounts to solving a
quadratic program. With this termination criterion only sub-motions
that bring the two polygons within a distance of $\varepsilon$ to each
other remain without a no-collision guarantee. Given the limited
accuracy of robot motion, this is probably desirable anyway.

\begin{algorithm}
  \caption{Collision detection; returns list of possibly colliding
    elements or empty list in case of no collision.}
  \label{alg:1}
  \begin{algorithmic}
    \Procedure{Collision}{$P, F, c, d$}
    \Comment{moving polygon $P$, fixed polygon $F$,\\\hfill motion $c$,
      recursion depth $d$}
    \If{$d = 0$}
      \State \Return parameter interval of $c$
      \Comment too many recursions
    \EndIf
    \State $g, R \gets \Call{BoundingBall}{c}$ \Comment center $g$, radius $R$
    \If{$R < R_{\min}$} \Comment{see Equation~\eqref{eq:11}}
      \State \Return parameter interval of $c$
      \Comment radius too small
    \EndIf
    \If{$\Call{Intersection}{P, F, g, R}$}
      \Comment{see Listing~\ref{alg:2} below}
      \State $c', c'' \gets \Call{Subdivide}{c}$
      \State \Return \Call{Collision}{$P, F, c' , d-1$}
      \State \Return \Call{Collision}{$P, F, c'', d-1$}
    \Else
      \State \Return $[]$ (empty list)
      \Comment{no collision}
    \EndIf
    \EndProcedure
  \end{algorithmic}
\end{algorithm}
\begin{algorithm}
  \caption{Intersection of fat moving polygon, determined by moving
    polygon and ball $\ball{g}{R}\subset\RSet^{12}$, and fixed
    polygon.}
  \label{alg:2}
  \begin{algorithmic}
    \Procedure{Intersection}{$P, F, g, R$}
    \Comment{moving polygon $P$, fixed polygon $F$,\\\hfill
      center $g$ and radius $R$}
    \Require{Metric data of plane of $P$ and edges of $P$ has been
      pre-computed and is accessible from within this procedure.}
    \State $P' \gets \Call{FatPolygon}{\mathrm{id}, R}$
    \Comment{fat polygon in ``zero position''}
    \State $F' \gets g^{-1}(F)$
    \For{each fat vertex $x$ of $P'$}
    \Comment{ball intersects polygon}
      \If{$x$ intersects $F'$}
        \State \Return true
      \EndIf
    \EndFor
    \For{each fat edge $e$ of $P'$}
      \Comment{line segment intersects hyperboloid}
      \For{each edge $f$ of $F'$}
        \If{$e$ intersects $f$}
          \State \Return true
        \EndIf
      \EndFor
    \EndFor
    \For{each vertex $y$ of $F'$}
      \Comment{point in hyperboloid, point projection in polygon}
      \If{$y$ lies in fat face of $P'$}
        \State \Return true
      \EndIf
    \EndFor
    \State \Return false
    \EndProcedure
  \end{algorithmic}
\end{algorithm}

The complete algorithm in pseudo-code is displayed in
Listings~\ref{alg:1} and \ref{alg:2}. The first call of procedure
\Call{Collision}{$P,F,c,d$} should have the complete vertex, edge
and face lists of both, moving and fixed polygon as arguments. A limit
on the recursion depth is good programming practice but, from a
theoretical point of view, not necessary.

\begin{remark}
  \label{rem:1}
  At this point we would like to mention one implementation detail we
  have not explained so far: It is possible that the radius $R$ of the
  bounding ball $\ball{g}{R}$ is so large that the fat edge or face is
  no longer bounded by a hyperboloid, compare Equations~\eqref{eq:6}
  and \eqref{eq:8}. The simplest method to circumvent this problem is
  to subdivide the motion and thus decrease the ball radius $R$. But
  it is also possible to exploit this behavior for simplified
  collision tests:
  \begin{itemize}
  \item If the inequality $R \ge \metdata_i$ holds for edge $e_i$, the
    fat edge is contained in the union of the two fat points over its
    end-points. Thus, we may be able to exclude a collision by testing
    for the intersection of a line segment with balls
    \cite[Section~5.3.2]{ericson05}.
  \item If $\max\{\metdata'_1,\metdata'_2\} \le R$, the fat vertices
    of the moving polygon envelope the complete fat polygon. The
    collision between $P$ and $F$ can be tested by solely checking
    possible intersections of $F$ with these vertex balls
    \cite[Section~5.2.8]{ericson05}.
  \item If either $\metdata'_1 \le R \le \metdata'_2$ or $\metdata'_2
    \le R \le \metdata'_1$ hold, the vertex spheres not necessarily
    envelope the toleranced position of $F$. By the Cauchy interlacing
    theorem \cite{hwang04}, there exist two straight lines in the
    plane of $P$ through the point of minimal distortion rate that
    divide the plane into four regions. In the interior of each of
    these regions, the situation is either as in the case
    $\metdata'_1,\metdata'_2 \le R$ or as in the case
    $\metdata'_1,\metdata'_2 \ge R$ and can be handled accordingly.
    Should $P$ overlap two or more regions we may subdivide~$P$.
  \end{itemize}
\end{remark}

\begin{remark}
  Finally, we would like to point out a possibility to considerably
  speed up the algorithm if the recursive subdivision and bounding
  ball construction for the motion $c(t)$ produces a \emph{nested}
  ball covering, that is, a ball covering where the balls after
  subdivision are contained in the ball before subdivision. In this
  case, it is sufficient to test only those elements for collision for
  which a collision at a preceding recursion level could not be
  excluded. It is possible to incorporate this in the presented
  algorithm by maintaining lists of vertices, edges and faces that
  still require collision testing. Methods for generating nested ball
  coverings are available \cite{quinlan94}. Their applicability in our
  algorithm depends on the type of motion we consider.
\end{remark}

\section{Examples and timings}
\label{sec:examples}

In this section, we describe a couple of examples in more detail. The
first group of examples considers a single fixed and a single moving
polygon that undergoes two different motions, one with collision, the
other without. The second group of examples is designed to demonstrate
the asymptotic behavior of our algorithm. We consider the collision
of two squares after subdivision into $n^2$ sub-squares.

The reader will notice that we always consider convex polygons. This
simplifies the implementation, in particular the collision test of
moving vertex and fixed face. Non-convex faces can always be
decomposed into convex parts. The impact on the overall performance
should be small.

For the computation of examples we used a personal computer with an
Intel Core 2 Duo CPU with 3000 MHz, 8 GB system memory and a 64-bit
Ubuntu 14.04. Programming language for this simple implementation is
C++, the code was compiled using g++ in version 4.8.2.

\subsection{Collision of one fixed and one moving polygon}
\label{sec:collision-one-one}

In this section we consider two different motions of a moving polygon
$P$ and its possible collisions with a fixed polygon $F$ (see
Figures~\ref{fig:showmotion-ni} and \ref{fig:showmotion-i}). The
polygons $P$ and $F$ are convex octagons. During the first motion, the
two polygons do not collide, during the second motion, they do. In
both cases, the motion is rational of degree eight. In fact, the
motions differ only by a constant translation. The Borel measure is
derived from point masses of equal weight in the vertices of the
moving polygon~$P$ and in a point over the barycenter of $P$ with
comparably small distance to the plane of $P$. This ensures that a
scaled copy of the inertia ellipsoid roughly captures the shape
of~$P$.

The motion is a rational Bézier curve $c(t)$ of degree eight in $\GA =
R^{12}$ with positive weights. We use the curve point $g \defeq
c(0.5)$ as center of a bounding ball $\ball{g}{R}$ and choose the
radius $R$ as the maximal distance to a control point. Should this
ball be too large to exclude collisions, we subdivide the curve by
means of de Casteljau's algorithm, compute bounding balls for both
parts and test again for collisions. This procedure we repeat until we
can exclude a collision or until we reach the termination criterion
described in \autoref{sec:collision-detection}.

It would be possible to use different strategies for computing
bounding balls, for example the minimal bounding ball of the control
polygon \cite{welzl91} or an approach inspired by \autoref{prop:2}. We
don't expect much difference but remark that neither approach is
guaranteed to produce a nested ball covering. Thus, we cannot profit
from pruning away non-colliding objects from previous stages. We did
not implement a special handling for balls of large radius as
explained in Remark~\ref{rem:1} but simply subdivide in this case.

In the first example, a collision can be excluded. The total
computation time minus the time for pre-processing is 0.1427 seconds.
The maximal recursion depth is eight but only 173 collision tests
polygon vs. fat polygon are performed. From
\autoref{fig:showmotion-ni} we can infer that excluding a collision in
this example is rather difficult because the moving polygon is always
``near'' the fixed polygon. If we translate the moving polygon by
about half of its diameter away from the fixed polygon, the
computation is roughly half as long.

In the second example, the two polygons interfere during a time
interval. We require a maximum of eight recursions and 0.0075 seconds
until we reach the termination criterion. The value of $\varepsilon$
is set to about one percent of the polygon diameters. A total of nine
collision tests polygon vs. fat polygon are performed.

\begin{figure}
  \begin{minipage}[b]{0.5\linewidth}
    \centering
    \includegraphics{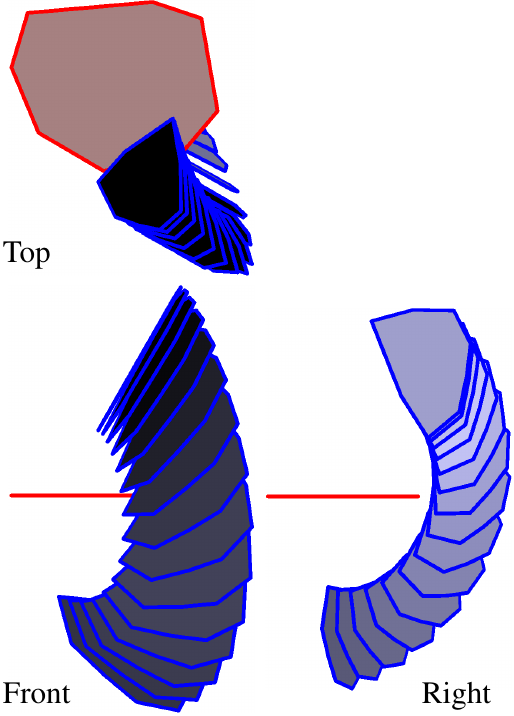}
    \caption{Motion of $P$ in Example 1}
    \label{fig:showmotion-ni}
  \end{minipage}%
  \begin{minipage}[b]{0.5\linewidth}
    \centering
    \includegraphics{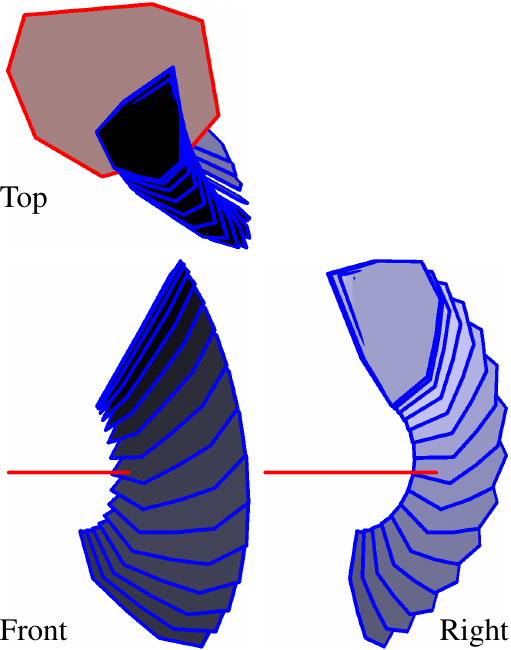}
    \caption{Motion of $P$ in Example~2}
    \label{fig:showmotion-i}
  \end{minipage}
\end{figure}

\subsection{Experiments on asymptotic behavior}
\label{sec:asymptotic-behavior}

Now we consider the potential collision of two polyhedral surfaces
that were obtained from two squares by subdivision into $n^2$
sub-squares. Thus, we have $n^4$ collision tests of moving vs. fixed
polygon. Of course this example is artificial but the algorithm does
not profit at all from the regular structure so that conclusions for
large polyhedral surfaces in general seem justified.

The two subdivided squares undergo three different rational motions of
degree eight: One with a collision, one with a near-collision and one
with a ``clear'' no-collision. The rotational component of these
motions is always the same, they differ only in a constant
translation. We run the algorithm for $1 \le n \le 10$ until we have a
no-collision guarantee or reach the termination criterion of the
previous examples. The recorded data is the total time (average over
100 runs), the total time per moving polyhedron (that is, collision of
one moving square with $n^2$ fixed squares) and the average time of
one collision test square vs. square.

Finally, we distinguish between two different strategies: We either
use one global mass distribution for the complete moving polyhedral
surface or we compute an individual mass distribution for each moving
square. The aim is to explore whether the additional computational
burden of the latter strategy is compensated by the faster detection
of a non-collision.

\begin{table}
  \centering
  \begin{tabular}{|r|rr|rr|rr|}
    \hline
    & \multicolumn{2}{c|}{Collision}
    & \multicolumn{2}{c|}{Near-Collision}
    & \multicolumn{2}{c|}{No-Collision} \\
    $n$ & total [s] & average [s] & total [s] & average [s] & total [s] & average [s] \\
    \hline
1 & 0.17 & 0.17 & 0.49 & 0.49 & 0.11 & 0.11 \\
2 & 1.58 & 0.10 & 2.43 & 0.15 & 1.12 & 0.07 \\
3 & 5.91 & 0.07 & 9.49 & 0.12 & 4.57 & 0.06 \\
4 & 16.07 & 0.06 & 25.78 & 0.10 & 12.17 & 0.05 \\
5 & 35.09 & 0.06 & 57.05 & 0.09 & 27.60 & 0.04 \\
6 & 67.41 & 0.05 & 110.34 & 0.09 & 54.82 & 0.04 \\
7 & 118.42 & 0.05 & 193.47 & 0.08 & 98.25 & 0.04 \\
8 & 194.16 & 0.05 & 318.39 & 0.08 & 164.00 & 0.04 \\
9 & 300.88 & 0.05 & 494.69 & 0.07 & 258.29 & 0.04 \\
10 & 448.09 & 0.04 & 738.22 & 0.07 & 388.67 & 0.04 \\
    \hline
  \end{tabular}\\
  \caption{Test data for collision, near-collision and no-collision
    with a global mass distribution}
  \label{tab:1}
\end{table}

\begin{table}
  \centering
  \begin{tabular}{|r|rr|rr|rr|}
    \hline
    & \multicolumn{2}{c|}{Collision}
    & \multicolumn{2}{c|}{Near-Collision}
    & \multicolumn{2}{c|}{No-Collision} \\
    $n$ & total [s] & average [s] & total [s] & average [s] & total [s] & average [s] \\
    \hline
1 & 0.17 & 0.17 & 0.49 & 0.49 & 0.11 & 0.11 \\
2 & 1.21 & 0.07 & 1.85 & 0.12 & 0.98 & 0.06 \\
3 & 3.68 & 0.04 & 5.99 & 0.07 & 3.22 & 0.04 \\
4 & 9.12 & 0.04 & 15.23 & 0.06 & 8.27 & 0.03 \\
5 & 19.01 & 0.03 & 32.59 & 0.05 & 18.74 & 0.03 \\
6 & 35.69 & 0.03 & 62.36 & 0.05 & 36.97 & 0.03 \\
7 & 61.59 & 0.03 & 108.64 & 0.04 & 66.01 & 0.03 \\
8 & 99.71 & 0.02 & 177.00 & 0.04 & 109.79 & 0.03 \\
9 & 153.22 & 0.02 & 274.48 & 0.04 & 172.29 & 0.03 \\
10 & 226.60 & 0.02 & 407.89 & 0.04 & 258.99 & 0.03 \\
    \hline
  \end{tabular}\\
  \caption{Test data for collision, near-collision and no-collision
    with individual mass distributions}
  \label{tab:2}
\end{table}

The data for total and average time is displayed in Table~\ref{tab:1}
for a global mass distribution and in Table~\ref{tab:2} for individual
mass distributions. The total time per moving polygon is visualized in
Figure~\ref{fig:time-per-moving-square} for a global mass distribution
and in Figure~\ref{fig:time-per-moving-square-2} for individual mass
distributions. In these figures, we plot the time per moving square as
a histogram over the moving squares. The height of each box is
proportional to the time for excluding a collision or reaching the
termination criterion without a no-collision guarantee. The latter
case is indicated by a colored box. We see that collisions or
near-collisions occur only at one corner and observe an increase of
computation time in the vicinity of this corner.

\begin{figure}
  \centering
  \newcommand{\igh}[1]{\begin{minipage}{0.18\linewidth}\centering\includegraphics[width=\linewidth]{hist_#1}\end{minipage}}
  \begin{tabular}{c|cccc}
                    & $n=1$           & $n=4$           & $n=7$           & $n=10$                                             \\
   \small collision &                 &                 &                 &                                                    \\
   \hline
   \small yes       & \igh{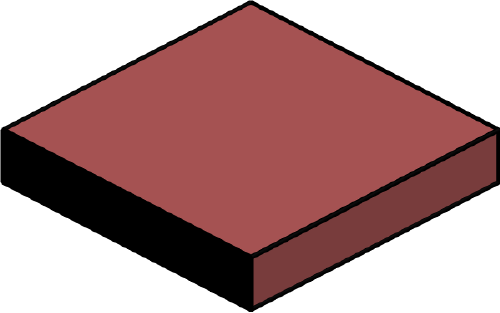}     & \igh{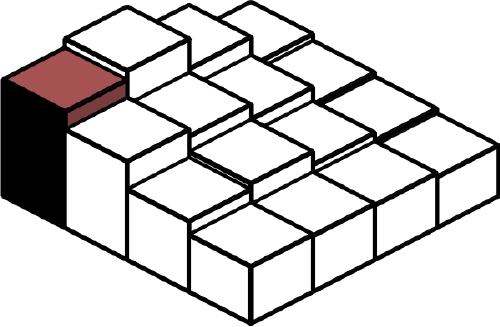}     & \igh{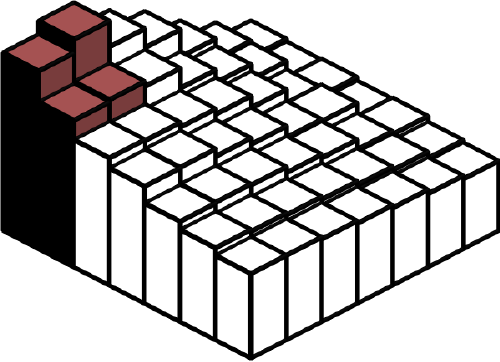}     & \igh{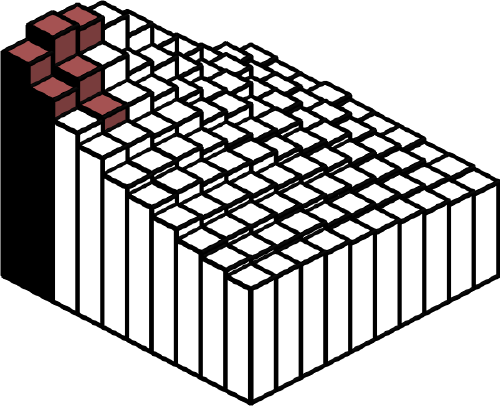}                                        \\
   \small near      & \igh{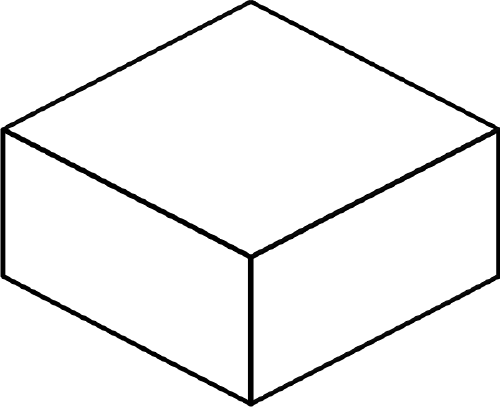} & \igh{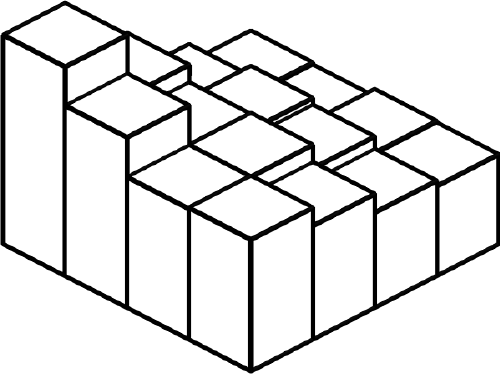} & \igh{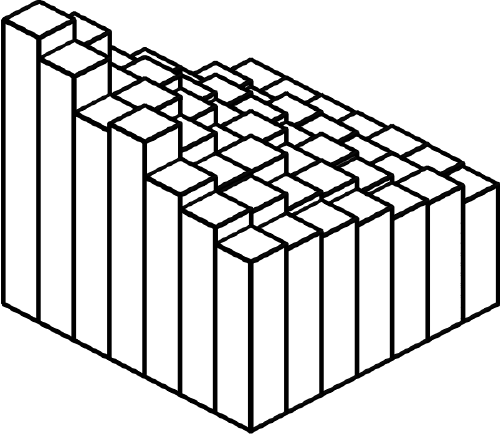} & \igh{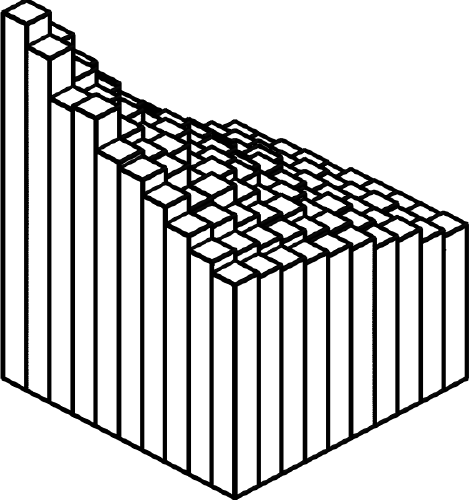}                                    \\
   \small no        & \igh{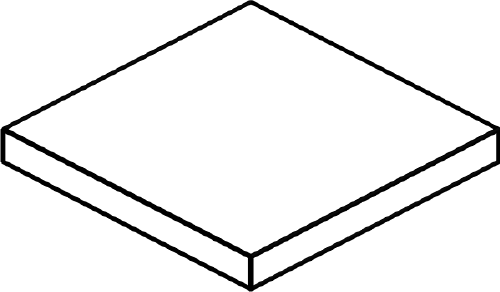}    & \igh{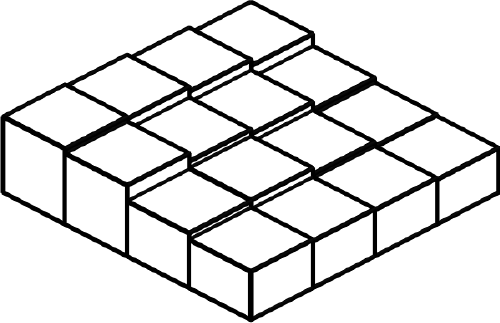}    & \igh{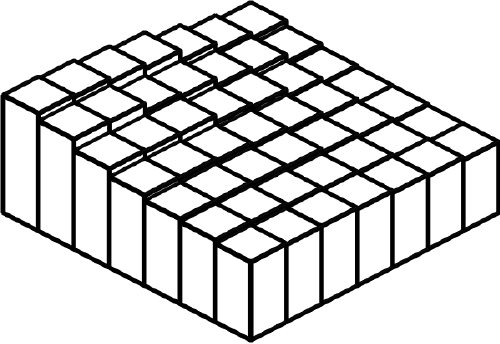}    & \igh{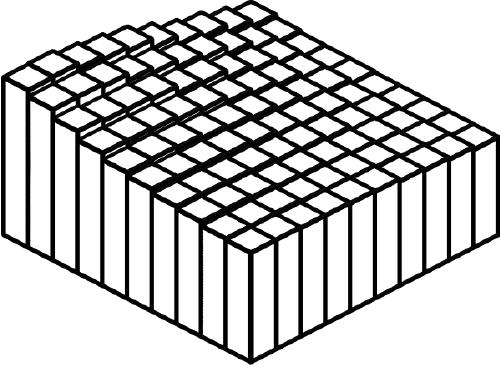}
  \end{tabular}
  \caption{Time per moving square with a global mass distribution}
  \label{fig:time-per-moving-square}
\end{figure}

\begin{figure}
  \centering
  \newcommand{\igm}[1]{\begin{minipage}{0.18\linewidth}\centering\includegraphics[width=\linewidth]{hist-massdist_#1}\end{minipage}}
  \begin{tabular}{c|cccc}
                    & $n=1$           & $n=4$           & $n=7$           & $n=10$                                                                        \\
   \small collision &                 &                 &                 &                                                                               \\
   \hline
   \small yes       & \igm{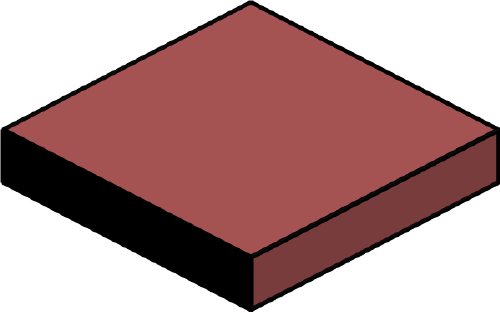}     & \igm{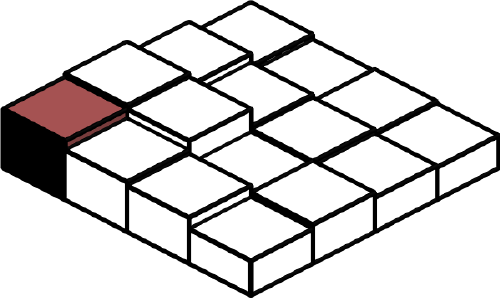}     & \igm{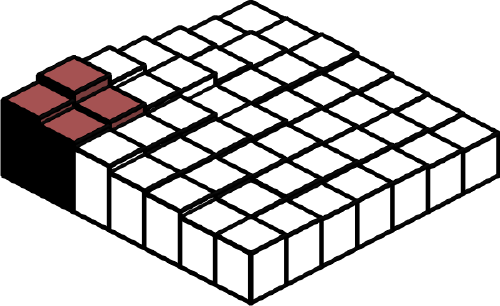}     & \igm{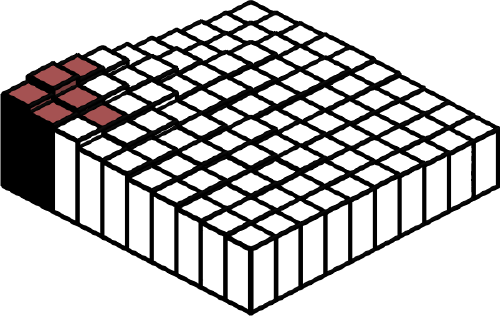}                                                                   \\
   \small near      & \igm{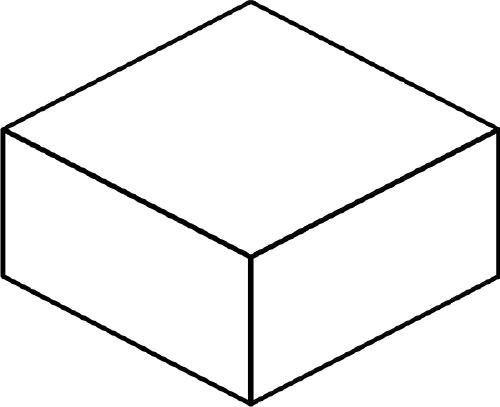} & \igm{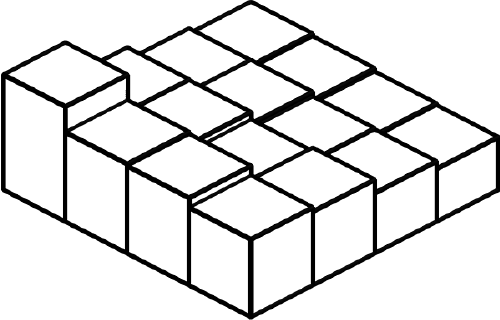} & \igm{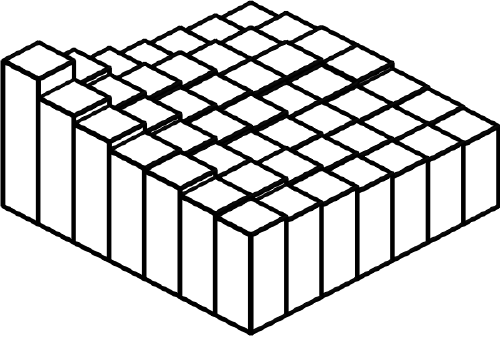} & \igm{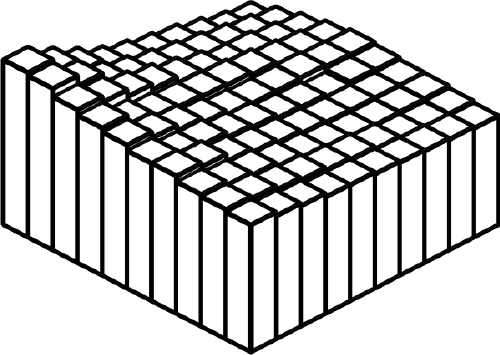}                                                               \\
   \small no        & \igm{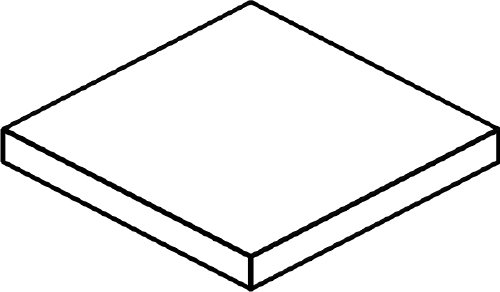}    & \igm{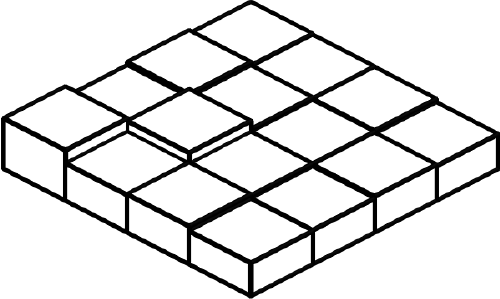}    & \igm{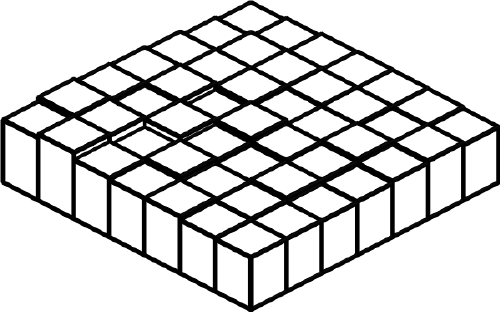}    & \igm{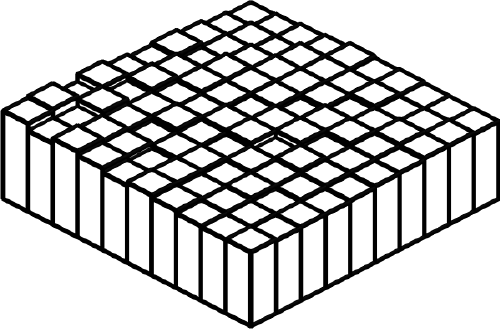}
  \end{tabular}
  \caption{Time per moving square with individual mass distributions}
  \label{fig:time-per-moving-square-2}
\end{figure}

The examples of this section allow to draw the following conclusions:
\begin{itemize}
\item The asymptotic behavior is linear in the total number $n^4$ of
  collision tests. This is to be expected since large $n$ means also a
  large number of no-collisions. The total time roughly equals the
  average time for the exclusion of a collision times the number of
  intersection tests for two polygons.
\item It is better to use individual mass distributions, in particular
  in the presence of collisions. The reason is that the time for
  computing the mass distribution is negligible but very often speeds
  up the detection of a non-collision. The ratio of total times is
  between approximately 1.5 in case of no collision and 2.0 in case of
  a collision. This is probably because many near-collisions in the
  latter case are recognized as no-collisions at a lower recursion
  depth.
\item The algorithm can greatly profit from a hierarchy of bounding
  polyhedra. This is nicely illustrated by
  Figures~\ref{fig:time-per-moving-square} and
  \ref{fig:time-per-moving-square-2}. In order to give a no-collision
  guarantee, it is sufficient to subdivide only the squares that, at a
  given recursion depth, correspond to colored boxes. At least in our
  examples, large parts of the polyhedral surface can be pruned away
  at low recursion depths.
\end{itemize}

\section{Concluding remarks}
\label{sec:concluding-remarks}

In this paper we exploited the simple orbit shape of points, lines,
planes and polygons with respect to balls $\ball{g}{R} \subset
\RSet^{12}$ of affine displacements for an efficient collision
detection algorithm. Emphasis is put on reliable no-collision
guarantee, simple implementation and fast detection of intervals with
no collision. A prototype implementation demonstrates that real-time
collision detection for reasonably simple polyhedral shapes undergoing
a rational motion is possible, provided there are few expected
collisions or near-collisions.

Our algorithm works with affine deformations while for many of the
envisaged applications just rigid body displacements are required.
Thus, one might wonder whether something could be gained by a
restriction to Euclidean motions. However, the tolerancing analysis in
\cite{schroecker05} gives no evidence for this believe. Linearization of
$\SE \subset \GA$ and subsequent error estimation comes at the cost of
more complicated orbit shapes. For example, the orbits of points with
respect to the intersection of a ball with $\SE$ are contained in
offsets to ellipsoids. Collision tests with these surfaces are so much
harder than with balls that one is certainly willing to accept more
recursions instead.

A necessary ingredient in our algorithm is the possibility to
efficiently generate bounding balls for curve segments. We already
remarked that for sufficiently short curves of small curvature one can
probably use the ball for which the curve end points are antipodal. It
would be good to have a more precise statement. Moreover, we emphasize
the need for efficient algorithms that generate nested ball
coverings for motions of technical relevance, in particular piecewise
rational motions. This is topic of future research.

\section*{Acknowledgments}
\label{sec:acknowledgments}

This research was supported by the Austrian Science Fund (FWF): 
P23831-N13.





\begin{thebibliography}{10}
\expandafter\ifx\csname url\endcsname\relax
  \def\url#1{\texttt{#1}}\fi
\expandafter\ifx\csname urlprefix\endcsname\relax\def\urlprefix{URL }\fi
\expandafter\ifx\csname href\endcsname\relax
  \def\href#1#2{#2} \def\path#1{#1}\fi

\bibitem{ericson05}
C.~Ericson, Real-Time Collision Detection, Morgan Kaufman, 2005.

\bibitem{hubbard95}
P.~M. Hubbard, Collision detection for interactive graphics applications, IEEE
  Trans. Visualization and Computer Graphics 1~(3) (1995) 218--230.

\bibitem{dingliana00}
J.~Dingliana, C.~O'Sullivan, Graceful degradation of collision handling in
  physically based animation, Eurographics 19~(3) (2000) 239--248.

\bibitem{belta02}
C.~Belta, V.~Kumar, An {SVD}-based projection method for interpolation on
  {SE(3)}, IEEE Transactions on Robotics and Automation 18~(3) (2002) 334--345.

\bibitem{chirikjian98}
G.~S. Chirikjian, S.~Zhou, Metrics on motion and deformation of solid models,
  ASME J. Mechanical Design 120~(2) (1998) 252--261.

\bibitem{hofer04:_motion_design}
M.~Hofer, H.~Pottmann, B.~Ravani, From curve design algorithms to the design of
  rigid body motions, The Visual Computer 20~(5) (2004) 279--297.

\bibitem{nawratil07}
G.~Nawratil, New performance indices for {6R} robots, Mech. Machine Theory
  42~(11) (2007) 1499--1511.

\bibitem{schroecker09:_evolving_four_bars}
H.-P. Schr\"ocker, B.~J\"uttler, M.~Aigner, Evolving four-bars for optimal
  synthesis, in: M.~Ceccarelli (Ed.), Proceedings of EUCOMES 08, Springer,
  2009, pp. 109--116.
\newblock \href {http://dx.doi.org/10.1007/978-1-4020-8915-2}
  {\path{doi:10.1007/978-1-4020-8915-2}}.

\bibitem{schroecker05}
H.-P. Schröcker, J.~Wallner, Curvatures and tolerances in the {Euclidean}
  motion group, Results Math. 47 (2005) 132--146.

\bibitem{quinlan94}
S.~Quinlan, Efficient distance computation between non-convex objects, in:
  Proceedings of IEEE Conference on Robotics and Automation, Vol.~4, San Diego,
  United States, 1994, pp. 3324--3329.
\newblock \href {http://dx.doi.org/10.1109/ROBOT.1994.351059}
  {\path{doi:10.1109/ROBOT.1994.351059}}.

\bibitem{greenspan96}
M.~Greenspan, K.~Burtnyk, Obstacle count independent real-time collision
  avoidance, in: Proceedings of IEEE International Conference on Robotics and
  Automation, Minneapolis, Minnesota, US, 1996, pp. 1073--1080.

\bibitem{martinez98}
B.~Martínez-Salvador, A.~P. del Pobil, M.~Pérez-Francisco, Very fast
  collision detection for practical motion planning part {I}: The spatial
  representation, in: Proceedings of IEEE Conference on Robotics and
  Automation, Leuven, Belgium, 1998, pp. 624--629.

\bibitem{weller09}
R.~Weller, G.~Zachmann, A unified approach for physically-based simulations and
  haptic rendering, in: Proceedings of the 2009 ACM SIGGRAPH Symposium on Video
  Games, Sandbox '09, ACM, New York, NY, USA, 2009, pp. 151--159.
\newblock \href {http://dx.doi.org/10.1145/1581073.1581097}
  {\path{doi:10.1145/1581073.1581097}}.

\bibitem{weller09b}
R.~Weller, G.~Zachmann, Inner sphere trees for proximity and penetration
  queries, in: 2009 Robotics: Science and Systems Conference (RSS), Seattle,
  WA, USA, 2009.

\bibitem{bernabeu01}
E.~J. Bernabeu, J.~Tornero, M.~Tomizuka, Collision prediction and avoidance
  amidst moving objects for trajectory planning applications, in: Proceedings
  of IEEE Conference on Robotics and Automation, Seoul, Korea, 2001, pp.
  3801--3806.

\bibitem{donev05}
A.~Donev, S.~Torquato, F.~H. Stillinger, Neighbor list collision-driven
  molecular dynamics for nonspherical hard particles: {I}. {Algorithmic
  details}, Journal of Computational Physics 202 (2005) 737--764.

\bibitem{donev05b}
A.~Donev, S.~Torquato, F.~H. Stillinger, Neighbor list collision-driven
  molecular dynamics for nonspherical hard particles: {II}. {Applications} to
  ellipses and ellipsoids, Journal of Computational Physics 202 (2005)
  765--793.

\bibitem{moore88}
M.~Moore, J.~Wilhelms, Collision detection and response for computer animation,
  Computer Graphics 22~(4) (1988) 289--298.

\bibitem{rouillier04}
F.~Rouillier, P.~Zimmerman, Efficient isolation of polynomial's real roots, J.
  Comput. Appl. Math. 162 (2004) 33--50.

\bibitem{hwang04}
S.-G. Hwang, Cauchy's interlace theorem for eigenvalues of hermitian matrices,
  Amer. Math. Monthly 111~(2) (2004) 157--159.

\bibitem{welzl91}
E.~Welzl, Smallest enclosing disks (balls and ellipsoids), in: New Results and
  New Trends in Computer Science, Vol. 555 of Lecture Notes in Computer
  Science, Springer, 1991, pp. 359--370.

\end{thebibliography}

\end{document}